\documentclass{llncs}

\usepackage[margin=1.1in]{geometry}

\usepackage{times}
\usepackage{color}
\usepackage{multirow}

\usepackage{amsmath}
\usepackage{amssymb}

\usepackage{cleveref}

\newcommand{\hide}[1]{}

\newcommand{\set}[1]{\left\{#1\right\}}
\newcommand{\ceil}[1]{\left\lceil #1 \right\rceil}
\renewcommand{\Pr}{\mathbb{P}}

\newcommand{\eps}{\varepsilon}
\newcommand{\calT}{\mathcal{T}}
\newcommand{\calS}{\mathcal{S}}

\renewcommand{\paragraph}[1]{\medskip\noindent{\bf #1\ }}

\author{Bernhard Haeupler\\MIT CSAIL\\Cambridge, MA, USA\\
\texttt{haeupler@mit.edu}%
 \and %
Fabian Kuhn\\ Dept.\ of Computer Science\\
University of Freiburg, Germany
\\\texttt{kuhn@informatik.uni-freiburg.de}}

\date{}

\begin{document}

\title{Lower Bounds on Information Dissemination in Dynamic Networks}
\author{Bernhard Haeupler\inst{1} \and Fabian Kuhn\inst{2}}
\institute{Computer Science and Artificial Intelligence Lab, MIT, USA\\
\email{haeupler@mit.edu} 
\and
Dept.\ of Computer Science, University of Freiburg, Germany
\email{kuhn@cs.uni-freiburg.de}}
\date{}

\maketitle

\begin{abstract}
    We study lower bounds on information dissemination in adversarial
    dynamic networks.
    Initially, $k$ pieces of information (henceforth called tokens)
    are distributed among $n$ nodes. The tokens need to be broadcast
    to all nodes through a synchronous network in which the topology
    can change arbitrarily from round to round provided that some
    connectivity requirements are satisfied.

    If the network is guaranteed to be connected in every round and
    each node can broadcast a single token per round to its neighbors,
    there is a simple token dissemination algorithm that manages to
    deliver all $k$ tokens to all the nodes in $O(nk)$
    rounds. Interestingly, in a recent paper, Dutta et al. proved an
    almost matching $\Omega(n+nk/\log n)$ lower bound for
    deterministic token-forwarding algorithms that are not allowed to
    combine, split, or change tokens in any way. In the present paper,
    we extend this bound in different ways.

    If nodes are allowed to forward $b\leq k$ tokens instead of only
    one token in every round, a straight-forward extension of the
    $O(nk)$ algorithm disseminates all $k$ tokens in time
    $O(nk/b)$. We show that for any randomized token-forwarding
    algorithm, $\Omega(n + nk/(b^2\log n\log\log n))$ rounds are
    necessary. If nodes can only send a single token per round, but we
    are guaranteed that the network graph is $c$-vertex connected in
    every round, we show a lower bound of $\Omega(nk/(c\log^{3/2}n))$, which
    almost matches the currently best $O(nk/c)$ upper bound. Further,
    if the network is $T$-interval connected, a notion that captures
    connection stability over time, we prove that $\Omega(n +
    nk/(T^2\log n))$ rounds are needed. The best known upper bound in
    this case manages to solve the problem in $O(n + nk/T)$
    rounds. Finally, we show that even if each node only needs to
    obtain a $\delta$-fraction of all the tokens for some $\delta\in
    [0,1]$, $\Omega(nk\delta^3/\log n)$ are still required.
\end{abstract}

\section{Introduction}

The growing abundance of (mobile) computation and communication
devices creates a rich potential for novel distributed systems and
applications. Unlike classical networks, often the resulting networks
and applications are characterized by a high level of churn and,
especially in the case of mobile devices, a potentially constantly
changing topology. Traditionally, changes in a network have been
studied as faults or as exceptional events that have to be tolerated
and possibly repaired. However, particularly in mobile applications,
dynamic networks are a typical case and distributed algorithms
have to properly work even under the assumption that the 
topology is constantly changing.

Consequently, in the last few years, there has been an increasing
interest in distributed algorithms that run in dynamic
systems. Specifically, a number of recent papers investigate the
complexity of solving fundamental distributed computations and
information dissemination tasks in dynamic networks, e.g.,
\cite{avin08,baumann09,clementi08,clementi09,cornejo12,LBarxiv,odell05,HK,KLO,KMO,KOSurvey}.
Particularly important in the context of this paper is the
synchronous, adversarial dynamic network model defined in
\cite{KLO}. While the network consists of a fixed set of participants
$V$, the topology can change arbitrarily from round to round, subject
to the restriction that the network of each round needs to be
connected or satisfy some stronger connectivity requirement.

We study lower bounds on the problem of disseminating a
bunch of tokens (messages) to all the nodes in a dynamic network as
defined in \cite{KLO}.\footnote{To be in line with \cite{KLO} and
    other previous work, we refer to the information pieces to be
    disseminated in the network as tokens.} Initially $k$ tokens are
placed at some nodes in the network. 
Time is divided into
synchronous rounds, the network graph of every round is connected, and
in every round, each node can broadcast one token to all its
neighbors. If in addition, all nodes know the size of the network $n$,
we can use the following basic protocol to broadcast all $k$ tokens to
all the nodes. The tokens are broadcast one after the other such that
for each token during $n-1$ rounds, every node that knows about the
token forwards it. Because in each round, there has to be an edge
between the nodes knowing the token and the nodes not knowing it, at
least one new node receives the token in every round and thus, after
$n-1$ rounds, all nodes know the token. Assuming that only one token
can be broadcast in a single message, the algorithm requires $k(n-1)$
rounds to disseminate all $k$ tokens to all the nodes.

Even though the described approach seems almost trivial, as long as we
do not consider protocols based on network coding, $O(nk)$ is the best
upper bound known.\footnote{In fact, if tokens and thus also messages
    are restricted to a polylogarithmic number of bits, even network
    coding does not seem to yield more than a polylog.\
    improvement \cite{AnalyzingNC,HK}.} In \cite{KLO}, a
\emph{token-forwarding algorithm} is defined as an algorithm that
needs to forward tokens as they are and is not allowed to combine or
change tokens in any way. Note that the algorithm above is a
token-forwarding algorithm. In a recent paper, Dutta et al.\ show that
for deterministic token-forwarding algorithms, the described simple
strategy indeed cannot be significantly improved by showing a lower
bound of $\Omega(nk/\log n)$ rounds \cite{LBarxiv}. Their lower bound
is based on the following observation. Assume that initially, every
node receives every token for free with probability $1/2$
(independently for all nodes and tokens). Now, with high probability,
whatever tokens the nodes decide to broadcast in the next round, the
adversary can always find a graph in which new tokens are learned
across at most $O(\log n)$ edges. Hence, in each round, at most
$O(\log n)$ tokens are learned.  Because also after randomly assigning
tokens with probability $1/2$, overall still roughly $nk/2$ tokens are
missing, the lower bound follows. 
We extend the
lower bound from \cite{LBarxiv} in various natural
directions. Specifically, we make the contributions listed in the
following. All our lower bounds hold for deterministic algorithms and
for randomized algorithms assuming a strongly adaptive adversary (cf.\
Section \ref{sec:model}). Our results are also summarized in
\Cref{table:bounds} which is discussed in Section \ref{sec:related}.

\paragraph{Multiple Tokens per Round:} Assume that instead of
forwarding a single token per round, each node is allowed to forward
up to $1<b\leq k$ tokens in each round. In the simple token-forwarding
algorithm that we described above, we can then forward a block of $b$
tokens to every node in $n-1$ rounds and we therefore get an
$O\big(\frac{nk}{b}\big)$ round upper bound.  We show that every
(randomized) token-forwarding algorithm needs at least $\Omega\big(n
+ \frac{nk}{b^2\log n\log\log n}\big)$ rounds. 

\paragraph{Interval Connectivity:} It is natural to assume that a
dynamic network cannot change arbitrarily from round to round and that
some paths remain stable for a while. This is formally captured by the
notion of interval connectivity as defined in \cite{KLO}. A network is
called $T$-interval connected for an integer parameter $T\geq 1$ if
for any $T$ consecutive rounds, there is a stable connected
subgraph. It is shown in \cite{KLO} that in a $T$-interval connected
dynamic network, $k$-token dissemination can be solved in $O\big(n
+\frac{nk}{T}\big)$ rounds. In this paper, we show that every
(randomized) token-forwarding algorithm needs at least $\Omega\big(n
+ \frac{nk}{T^2\log n}\big)$ rounds. 

\paragraph{Vertex Connectivity:} If instead of merely requiring that
the network is connected in every round, we assume that the network is
$c$-vertex connected in every round for some $c>1$, we can also obtain
a speed-up. Because in a $c$-vertex connected graph, every vertex cut
has size at least $c$, if in a round all nodes that know a token $t$
broadcast it, at least $c$ new nodes are reached. The basic
token-forwarding algorithm thus leads to an $O\big(\frac{nk}{c}\big)$
upper bound. We prove this upper bound tight up to a small factor by
showing an $\Omega\big(\frac{nk}{c\log^{3/2}n}\big)$ lower bound.

\paragraph{\boldmath$\delta$-Partial Token Dissemination:} 
Finally 
we consider the basic model, but relax the requirement on the problem
by requiring that every node needs to obtain only a $\delta$-fraction
of all the $k$ tokens for some parameter $\delta\in[0,1]$. We show
that even then, at least $\Omega\big(\frac{nk\delta^3}{\log n}\big)$
rounds are needed. This also has implications for algorithms that use
forward error correcting codes (FEC) to forward coded packets instead
of tokens.  We show that such algorithms still need at least
$\Omega\big(n + k\big(\frac{n}{\log n}\big)^{1/3}\big)$ rounds until
every node has received enough coded packets to decode all $k$
tokens.  


\section{Related Work}
\label{sec:related}

\begin{table}[t]
\centering
\begin{tabular}{ | l || c | c | c | } \hline
    & \ TF Alg. \cite{KLO} \  & \ NC Alg. \cite{AnalyzingNC,HK,HM}\  & \ TF Lower Bound \ \\ \hline \hline 
\multirow{2}{*}{always connected} & \multirow{2}{*}{$nk$ } 	&  \
\ \ \ \ \ \ \ \  $O(\frac{nk}{\log n})$ \ \ \ \ \ \ \small (1)  &  \
\ \ \ \ \ \ \ \   $\Omega(n\log k)$ \ \ \ \ \ \ \ \ \small\cite{KLO}\\ 
 & &  \ \ \ \ \ \ \ \ \  $O(n+k)$ \ \ \ \ \small (2) &  \ \ \ \ \ \ \ \ \   $\Omega(\frac{nk}{\log n})$ \ \ \ \ \ \ \ \ \ \ \ \ \small \cite{LBarxiv} \\ \hline

$T$-interval conn.$^+$  	& \multirow{2}{*}{\ \ \ \ \ \ \ \ 
    $\ \,\frac{nk}{T}$  \ \small (+,*)  } 	& \ \ \ \ \ \
$\ \ \,\approx O(n+\frac{nk}{T^2})$ \, \small (1,*)& \multirow{2}{*}{\ \ \ \
    {\boldmath$\Omega(\frac{nk}{T^2 \log n})$} \ \ \small (+) } \\ 

$T$-stability$^*$	&  	& \ \ \ \ \ \ \ \ \ \ 	$O(n + k)$ \ \ \ \ \small (2) &   \\ \hline

\multirow{2}{*}{always $c$-connected}  			& \multirow{2}{*}{$\frac{nk}{c}$ } 	&  \ \ \ \ \ \ \ \ \ \ 	$O(\frac{nk}{c\log n})$ \ \ \,\, \small  (1) &  {\boldmath$\Omega(\frac{nk}{c\log^{3/2} n})$} \\ 

			& &	 \ \ \ \ \ \ \ \ \ \ $O(\frac{n +
                            k}{c})$ \ \ \ \ \ \ \ \small (2)& {
                            {\boldmath$\Omega(\frac{nk}{c^2\log n})$}
                            }  \\ \hline

\multirow{2}{*}{$b$-token packets}  			&
\multirow{2}{*}{$\frac{nk}{b}$ } 	& \ \ \ \ \ \  	$\approx
O(\frac{nk}{b^2 \log n})$ \ \ \ \small (1)&  \multirow{2}{*}{\boldmath$\Omega(\frac{nk}{b^2 \log n \log \log n})$}  \\ 

  			&  	&\ \ \ \ \ \ \ \ \ \ 	$O(n + \frac{k}{b})$ \ \ \ \ \small (2)&   \\ \hline

\multirow{2}{*}{$\delta$-partial token diss.} & \multirow{2}{*}{$\delta nk$} 	&\ \ \ \ \ \ \ \ \ \ 	 $O(\frac{\delta nk}{\log n})$ \ \ \ \ \ \ \ \small (1)  &  \multirow{2}{*}{ {\boldmath$\Omega(\frac{\delta^3 nk}{\log n})$} }\\ 
 & & \ \ \ \ \ \ \ \ \ \ $O(n+\delta k)$ \ \,\ \small (2) & \\ \hline
\end{tabular}\\[1mm]
\caption{Upper and lower bounds for token forwarding (TF) algorithms
    and network coding (NC) based solutions (bounds in bold are proven
    in this paper). All TF algorithms are
    distributed and deterministic while all lower bounds are for
    centralized randomized algorithms and a strongly adaptive
    adversary. The NC algorithms work either in the distributed
    setting against a (standard) adaptive adversary (1) or in the
    centralized setting against a strongly adaptive adversary (2).} 
\label{table:bounds}
\end{table}

As stated in the introduction, we use the network model introduced in
\cite{KLO}. That paper studies the complexity of computing basic
functions such as counting the number of nodes in the network, as well
as the cost the token dissemination problem that we investigate in the
present paper. Previously, some basic results of the same kind were
also obtained in \cite{odell05} for a similar network model. 

The token dissemination problem as studied here is first considered in
\cite{KLO} in a dynamic network setting. The paper gives a variant of
the distributed $O(nk)$ token-forwarding algorithm for the case when
the number of nodes $n$ is not known. It is also shown that
$T$-interval connectivity and always $c$-vertex connectivity are
interesting parameters that speed up the solution by factors
of $\Theta(T)$ and $\Theta(c)$, respectively. In addition, \cite{KLO}
gives a first $\Omega(n\log k)$ lower bound for token-forwarding
algorithms in the centralized setting we study in the present
paper. That lower bound is substantially improved in \cite{LBarxiv},
where an almost tight $\Omega(nk/\log n)$ lower bound is proven. As
the lower bound from \cite{LBarxiv} is the basis of our results, we
discuss it in detail in Section \ref{sec:original}.

The fastest known algorithms for token dissemination in dynamic
networks are based on random linear network coding. There, tokens are
understood as elements (or vectors) of a finite field and in every
round, every node broadcasts a random linear combination of the
information it possesses. In a centralized setting, the overhead for
transmitting the coefficients of the linear combination can be
neglected. For this case, it is shown in \cite{AnalyzingNC} that in
always connected dynamic networks, $k$ tokens can be disseminated in
optimal $O(n+k)$ time. If messages are large enough to store $b$
tokens, this bound improves to again optimal $O(n+k/b)$ time. It is
also possible to extend these results to always $c$-connected networks
and to the partial token dissemination problem. Note that one possible
solution for $\delta$-partial token dissemination is to solve regular
token dissemination for only $\delta k$ tokens. If the overhead for
coefficients is not neglected, the best known upper bounds are given
in \cite{HK}. The best bounds for tokens of size $O(\log n)$, as well
as the upper and lower bounds for the other scenarios are listed in
\Cref{table:bounds}. The given bound for always $c$-vertex connected
networks is not proven in \cite{HK}, it can however be obtained with
similar techniques. Note also that instead of $T$-interval
connectivity, \cite{HK} considers a somewhat stronger assumption
called $T$-stability. In a $T$-stable network, the network remains
fixed for intervals of length $T$.

Apart from token dissemination and basic aggregation tasks, other
problems have been considered in the same or similar adversarial
dynamic network models. In \cite{KMO}, the problem of coordinating
actions in time is studied for always connected dynamic networks. In a
recent paper, bounds on what can be achieved if the network is not
always connected are discussed in \cite{cornejo12}. For a model where
nodes know their neighbors before communicating, \cite{avin08} studies
the time to do a random walk if the network can change
adversarially. Further, the problem of gradient clock synchronization
has been studied for an asynchronous variant of the model
\cite{clocksync}. In addition, a number of papers investigate a radio
network variant of essentially the dynamic network model studied here
\cite{anta10,clementi09,dualgraphs}. Another line of
research looks at random dynamic networks that result from some Markov
process, e.g., \cite{baumann09,clementi09b,clementi12}. Mostly these
papers analyze the time required to broadcast a single message in the
network. 
For a more thorough
discussion of related work, we refer to a recent survey
\cite{KOSurvey}.

\section{Model and Problem Definition}
\label{sec:model}

In this section we introduce the dynamic network model and the token
dissemination problem.

\paragraph{Dynamic Networks:}
We follow the dynamic network model of \cite{KLO}: A dynamic network
consists of a fixed set $V$ of $n$ nodes and a dynamic edge set
$E:\mathbb{N}\to 2^{\set{\set{u,v}|u,v\in V}}$. Time is divided into
synchronous rounds so that the network graph of round $r\geq 1$ is
$G(r)=(V,E(r))$. We use the common assumption that round $r$ starts at
time $r-1$ and it ends at time $r$. In each round $r$, every node
$v\in V$ can send a message to all its neighbors in $G(r)$. Note that
we assume that $v$ has to send the same message to all neighbors,
i.e., communication is by local broadcast. Also, we assume that at the
beginning of a round $r$, when the messages are chosen, nodes are not
aware of their neighborhood in $G(r)$. We typically assume that the
message size is bounded by the size of a fixed number of tokens.

We say that a dynamic network $G=(V,E)$ is \emph{always $c$-vertex
    connected} iff $G(r)$ is $c$-vertex connected for every round
$r$. If a network $G$ is always $1$-vertex connected, we also say that
$G$ is \emph{always connected}. Further, we use the definition for
interval connectivity from \cite{KLO}. A dynamic network is
$T$-interval connected for an integer parameter $T\geq 1$ iff the
graph $\big(V,\bigcap_{r'=r}^{r+T-1} E(r')\big)$ is connected for
every $r\geq 1$. Hence, a graph is $T$-interval connected iff there is
a stable connected subgraph for every $T$ consecutive rounds. Note we
do not assume that nodes know the stable subgraph. Also note that a
dynamic graph is $1$-interval connected iff it is always connected.

For our lower bound, we assume randomized algorithms and a
\emph{strongly adaptive adversary} which can decide on the network
$G(r)$ of round $r$ based on the complete history of the network up to
time $r-1$ as well as on the messages the nodes send in round
$r$. Note that the adversary is stronger than the more typical
\emph{adaptive adversary} where the graph $G(r)$ of round $r$ is
independent of the random choices that the nodes make in round
$r$.

\paragraph{The Token Dissemination Problem:}
We prove lower bounds on the following \emph{token dissemination
    problem}. There are $k$ tokens initially distributed among the
nodes in the network (for simplicity, we assume that $k$ is at most
polynomial in $n$). We consider \emph{token-forwarding algorithms} as
defined in \cite{KLO}. In each round, every node is allowed to
broadcast $b\geq 1$ of the tokens it knows to all neighbors. Except
for Section \ref{sec:multipletokens}, we assume that $b=1$. No other
information about tokens can be sent, so that a node $u$ knows exactly
the tokens $u$ kept initially and the tokens that were included in
some message $u$ received. In addition, we also consider the
\emph{$\delta$-partial token dissemination} problem. Again, there are
$k$ tokens that are initially distributed among the nodes in the
network. But here, the requirement is weaker and we only demand that
in the end, every node knows a $\delta$-fraction of the $k$ tokens for
some $\delta\in(0,1]$.

We prove our lower bounds for centralized algorithms where a central
scheduler can determine the messages sent by each node in a round $r$
based on the initial state of all the nodes before round $r$. Note
that lower bounds obtained for such centralized algorithms are
stronger than lower bounds for distributed protocols where the message
broadcast by a node $u$ in round $r$ only depends on the initial
state of $u$ before round $r$.

\section{Lower Bounds}
\label{sec:lowerbounds}

\subsection{General Technique and Basic Lower Bound Proof}
\label{sec:original}

We start our description of the lower bound by outlining the basic
techniques and by giving a slightly polished version of the lower
bound proof by Dutta et al.~\cite{LBarxiv}. For the discussion here,
we assume that in each round, each node is allowed to broadcast a
single token, i.e., $b=1$.

In the following, we make the standard assumption that round $r$ lasts
from time $r-1$ to time $r$. For each node, we maintain two sets of
tokens. For a time $t\geq 0$ and a node $u$, let $K_u(t)$ be the set
of tokens known by node $u$ at time $t$. In addition the adversary
determines a token set $K_u'(t)$ for every node, where
$K_u'(t)\subseteq K_u'(t+1)$ for all $t\geq 0$. The sets $K_u'(t)$ are
constructed such that under the assumption that each node $u$ knows
the tokens $K_u(t)\cup K_u'(t)$ at time $t$, in round $t+1$, overall
the nodes cannot learn many new tokens. Specifically, we define a
potential function $\Phi(t)$ as follows:
\begin{equation}
    \label{eq:potential}
    \Phi(t) := \sum_{u\in V} \left|K_u(t) \cup K_u'(t)\right|.
\end{equation}
Note that for the token dissemination problem to be completed at time
$T$ it is necessary that $\Phi(T) = nk$. Assume that at the beginning,
the nodes know at most $k/2$ tokens on average, i.e., $\sum_{u\in
    V}|K_u(0)|\leq nk/2$. For always connected dynamic graphs, we will
show that there exists a way to choose the $K'$-sets such that
$\sum_{u\in V}|K_u'(0)|<0.3nk$ and that for every choice of the
algorithm, a simple greedy adversary can ensure that the potential
grows by at most $O(\log n)$ per round. We then have $\Phi(0)\leq
0.8nk$ and since the potential needs to grow to $nk$, we get an
$\frac{0.2nk}{O(\log n)}$ lower bound.

In each round $r$, for each node $u$, an algorithm can decide on a
token to send. We denote the token sent by node $u$ in round $r$ by
$i_u(r)$ and we call the collection of pairs $(u,i_u(r))$ for nodes
$u\in V$, the token assignment of round $r$. Note that because a node
can only broadcast a token it knows, $i_u(r)\in K_u(r-1)$ needs to
hold. However, for most of the analysis, we do not make use of
this fact and just consider all the $k$ possible pairs $(u,i_u(r))$
for a node $u$.

If the graph $G(r)$ of round $r$ contains the edge $\set{u,v}$, $u$ or
$v$ learns a new token if $i_v(r)\not\in K_u(r-1)$ or if
$i_u(r)\not\in K_v(r-1)$. Moreover, the edge $\set{u,v}$ contributes
to an increase of the potential function $\Phi$ in round $r$ if
$i_v(r)\not\in K_u(r-1)\cup K_u'(r-1)$ or if $i_u(r)\not\in
K_v(r-1)\cup K_v'(r-1)$. We call an edge $e=\set{u,v}$ \emph{free} in
round $r$ iff the edge does not contribute to the potential difference
$\Phi(r)-\Phi(r-1)$. In particular, this implies that an edge is free
if
\begin{equation}\label{eq:freeedge}
    \big(i_u(r)\in K'_v(r-1) \land i_v(r)\in K'_u(r-1)\big) \lor \big(i_u(r)=i_v(r)\big).
\end{equation}

To construct the $K'$-sets we use the probabilistic method. More
specifically, for every token $i$ and all nodes $u$, we independently
put $i \in K'_u(0)$ with probability $p=1/4$. The following lemma shows
that then only a small number of non-free edges are required in every
graph $G(r)$.

\begin{lemma}[adapted from \cite{LBarxiv}]\label{lemma:freeedges}
    If each set $K'_u(0)$ contains each token $i$ independently with
    probability $p=1/4$, for every round $r$ and every token
    assignment $\set{(u,i_u(r))}$, the graph $F(r)$ induced by all free edges
    in round $r$ has at most $O(\log n)$ components with probability
    at least $3/4$.
\end{lemma}
\begin{proof}
    Assume that the graph $F(r)$ has at least $s$ components for some
    $s\geq 1$. $F(r)$ then needs to have an independent set of size
    $s$, i.e., there needs to be a set $S\subseteq V$ of size $|S|\geq
    s$ such that for all $u,v\in S$, the edge $\set{u,v}$ is not free
    in round $r$. Using \eqref{eq:freeedge} and the fact that
    $K_u'(0)\subseteq K_u'(t)$ for all $u$ and $t\geq 0$, an edge
    $\set{u,v}$ is free in round $r$ if $i_u(r)\in K_v'(0)$ and
    $i_v(r)\in K_u'(0)$ or if $i_u(r)=i_v(r)$. 

    To argue that $s$ is always small we use a union bound over all
    $\binom{n}{s} < n^s$ ways to choose a set of $s$ nodes and all at
    most $k^s$ ways to choose the tokens to be sent out by these
    nodes. Note that since two nodes sending out the same token induce
    a free edge, all tokens sent out by nodes in $S$ have to be
    distinct. Furthermore, for any pair of nodes $u,v \in S$ there is
    a probability of exactly $p^2$ for the edge $\set{u,v}$ to be free
    and this probability is independent for any pair $u',v'$ with
    $\set{u',v'}\neq \set{u,v}$ because nodes in $S$ send distinct
    tokens. The probability that all $\binom{s}{2} > s^2/4$ node pairs
    of $S$ are non-free is thus exactly $(1 - p^2)^{\binom{s}{2}} <
    e^{-p^2 s^2/4}$. If $s = 12 p^{-2} \ln nk > 4p^{-2}(\ln nk + 2)$
    (assuming $\ln(nk)>1$), the union bound $(nk)^s e^{-p^2 s^2/4}$ is
    less than $1/4$ as desired. This shows that there is a way to
    choose the sets $K'_u(0)$ such that the greedy adversary always
    chooses a topology in which the graph $F(r)$ induced by all free
    edges has at most 
    $ 2s\leq 24 p^{-2} \ln nk = O(\log n)$
    components.
\hspace*{\fill}\qed\end{proof}

Based on Lemma \ref{lemma:freeedges}, the lower bound from
\cite{LBarxiv} now follows almost immediately.

\begin{theorem}\label{thm:original}
    In an always connected dynamic network with $k$ tokens in
    which nodes initially know at most $k/2$ tokens on average,
    any centralized token-forwarding algorithm takes at least
    $\Omega\big(\frac{nk}{\log n}\big)$ rounds to disseminate all
    tokens to all nodes.
\end{theorem}
 \begin{proof}
     By independently including each token with probability $1/4$ in
     each of the sets $K_u'(0)$, we have that $\sum_u |K'_u| < 0.3 nk$
     with probability at least $3/4$ (for sufficiently large
     $nk$). Further, by Lemma \ref{lemma:freeedges}, with probability
     at least $3/4$, we obtain sets $K_u'(0)$ such that the potential
     can only grow by $O(\log n)$ in every round. Hence, there exists
     set $K'_u(0)$ such that the initial potential is at most $0.8nk$
     and in each round, the potential function does not grow by more
     than $O(\log n)$. As in the end the potential function has to
     reach $nk$, the claim then follows.
 \hspace*{\fill}\qed\end{proof}

\subsection{Partial Token Dissemination}
\label{app:partial}

We conclude our discussion of generalizations of the basic lower bound
proof of Section \ref{sec:original} by showing two relatively simple
results concerning partial token dissemination and a related problem.

\begin{theorem}\label{thm:deltaLB}
    For any $\delta > 0$, suppose an always connected dynamic network
    with $k$ tokens in which nodes initially know at most $\delta k/2$
    tokens on average.  Then, any centralized token-forwarding
    algorithm requires at least $\Omega(\frac{nk\delta^3}{\log(nk)})$
    rounds to solve $\delta$-partial token dissemination.
\end{theorem}
\begin{proof}
    The proof is analogous to the proof in Section
    \ref{sec:original}. Again, we construct the $K'$-sets using the
    probabilistic method. Here, we include every token in every set
    $K'_u(0)$ with probability $p=\delta/4$.  For sufficiently large
    $n$, we then get that $\Phi(0) < 0.8 \delta k n$ with probability
    at least $3/4$. A potential of at least $\Phi(T) \geq \delta nk$
    is needed to terminate at time $T$.  Following the same proof as
    for \Cref{lemma:freeedges}, there exists $K'$-sets such that in
    each round the potential increases by at most $24 p^{-2} \ln nk =
    O(\delta^{-2} \log nk)$ which implies a $\frac{\delta
        nk}{O(\delta^{-2} \log
        nk)}=\Omega\big(\frac{nk\delta^3}{\log(nk)}\big)$ lower bound.
\hspace*{\fill}\qed\end{proof}

\subsubsection*{Token Dissemination Based on Forward Error Correction}

Let us now consider an interesting special case where initially one
node knows all the tokens. In this situation, a simple way of applying
coding for token dissemination is to use forward error correcting
codes (FEC). From the $k$ tokens, such a code generates a large number
of code words (of essentially the same length as one of the message),
so that getting any $k$ code words allows to reconstruct all the $k$
messages.

\begin{theorem}
    Any token dissemination algorithm as described above takes at
    least $\Omega(n + k (\frac{n}{\log (nk)})^{1/3})$ rounds to
    disseminate $k$ tokens.
\end{theorem}
\begin{proof}
    Let $T$ be the time in which the FEC-based token dissemination
    algorithm terminates. The total number of different FEC messages
    sent is at most $T$ and every node needs to receiver at least a
    $\delta = k / T$ fraction of these messages. From
    \Cref{thm:deltaLB} we thus get that $T = \Omega(\frac{nT
        (k/T)^3}{\log (nk)})$ which leads to $T^3 > \Omega(k^3
    \frac{n}{\log n})$. As the network can be a static network of
    diameter linear in $n$, $\Omega(n)$ is clearly also a lower bound
    on the time needed. Together, the two bounds imply the claim of
    the theorem.
\hspace*{\fill}\qed\end{proof}

\subsection{Sending Multiple Tokens per Round}
\label{sec:multipletokens}

In this section we show that it is possible to extend the lower bound
to the case where nodes can send out $b>1$ tokens in each round. Note
that it is a priori not clear that this can be done as for instance
the related $\Omega(n\log k)$ lower bound of \cite{KLO} breaks down
completely if nodes are allowed to send two instead of one tokens in
each round.

In order to prove a lower bound for $b>1$, we generalize the notion of
free edges. Let us first consider a token assignment for the case
$b>1$. Instead of sending a single token $i_u(r)$, each node $u$ now
broadcasts a set $I_u(r)$ of at most $b$ tokens in every round
$r$. Analogously to before, we call the collection of pairs
$\big(u,I_u(r)\big)$ for $u\in V$, the token assignment of round
$r$. We define the weight of an edge in round $r$ as the amount the
edge contributes to the potential function growth in round $r$. Hence,
the weight $w(e)$ of an edge $e=\set{u,v}$ is defined as
\begin{equation}
    \label{eq:edgeweight}
    w(e) := \left|I_v(r)\setminus (K_u(r\!-\!1)\cup K_u'(r\!-\!1))\right| + 
    \left|I_u(r)\setminus (K_v(r\!-\!1)\cup K_v'(r\!-\!1))\right|.
\end{equation}
As before, we call an edge $e$ with weight $w(e)=0$ free. Given the
edge weights and the potential function as in Section
\ref{sec:original}, a simple possible strategy of the adversary works
as follows. In each round, the adversary connects the nodes using an
MST w.r.t.\ the weights $w(e)$ for all $e\in {V\choose 2}$. The total
increase of the potential function is then upper bounded by the weight
of the MST.

For the MST to contain $\ell$ or more edges of weight at least $w$,
there needs to be set $S$ of $\ell+1$ nodes such that the weight of
every edge $\set{u,v}$ for $u,v\in S$ is at least $w$. The following
lemma bounds the probability for this to happen, assuming that the
$K'$-sets are chosen randomly such that every token $i$ is contained
in every set $K'_u(0)$ with probability $p=1-\eps/(4eb)$ for some
constant $\eps>0$. 

\begin{lemma}\label{lemma:cliques}
    Assume that each set $K_u'(0)$ contains each token independently
    with probability $1-\eps/(4e b)$. Then, for every token assignment
    $(u,I_u(r))$, there exists a set $S$ of size $\ell+1$ such that
    all edges connecting nodes in $S$ have weight at least $w$ with
    probability at most
    \[
    \exp\left((\ell+1)\cdot\left(
            \ln n + b\ln k + \ell+1 - 
            \frac{\ell w}{12}\ln\left(\frac{w}{\eps}\right)
        \right)\right).
    \]
\end{lemma}
\begin{proof}
    Consider an arbitrary (but fixed) set of nodes $v_0,\dots,v_\ell$
    and a set of token sets $T_0,\dots,T_\ell$ (we assume that the
    token assignment contains the $\ell+1$ pairs $(v_i,T_i)$). We
    define $\mathcal{E}_i$ to be the event that $\big|\bigcup_{j\neq
        i} T_j\setminus K_{v_i}'(0)\big|> \ell w /4$. Note that
    whenever $|K_{v_i}\cup K_{v_i}'|$ grows by more than $\ell w/4$,
    the event $\mathcal{E}_i$ definitely happens. In order to have
    $|T_j\setminus K_{v_i}'(0)|+|T_i\setminus K_{v_j}'(0)|\geq w$ for
    each $i\neq j$, at least $(\ell+1)/3$ of the events
    $\mathcal{E}_i$ need to occur. Hence, for all edges
    $\set{v_i,v_j}$, $i,j\in\set{0,\dots,\ell}$, to have weight at
    least $w$, at least $(\ell+1)/3$ of the events $\mathcal{E}_i$
    have to happen. As the event $\mathcal{E}_i$ only depends on the
    randomness used to determine $K_{v_i}'(0)$, events $\mathcal{E}_i$
    for different $i$ are independent. The number of events
    $\mathcal{E}_i$ that occur is thus dominated by a binomial random
    variable $\mathrm{Bin}\big(\ell+1,\max_i\Pr[\mathcal{E}_i]\big)$
    variable with parameters $\ell+1$ and
    $\max_i\Pr[\mathcal{E}_i]$. The probability $\Pr[\mathcal{E}_i]$
    for each $i$ can be bounded as follows:
    \[
    \Pr[\mathcal{E}_i] \leq {\ell b\choose \ell w/4} \cdot
    \left(\frac{\eps}{4e b}\right)^{\ell w/4} \leq
    \left(\frac{4e\ell b}{\ell w}\right)^{\ell
        w/4}\cdot\left(\frac{\eps}{4e b}\right)^{\ell w /4} =
    \left(\frac{\eps}{w}\right)^{\ell w /4}.
    \]
    Let $X$ be the number of events $\mathcal{E}_i$ that occur. We
    have
    \[
    \Pr\left[X\geq \frac{\ell+1}{3}\right] \leq 
    {\ell+1\choose (\ell+1)/3}\cdot
    \left(\frac{\eps}{w}\right)^{\frac{\ell w}{4}\cdot \frac{\ell+1}{3}}
    \leq
    2^{\ell+1}\cdot
    \left(\frac{\eps}{w}\right)^{\frac{\ell w}{4}\cdot \frac{\ell+1}{3}}.
    \]
    The number of possible ways to choose $\ell+1$ nodes and assign a
    set of $b$ tokens to each node is
    \[
    {n\choose \ell+1}\cdot {k\choose b}^{\ell+1} \leq 
    \left(n k^b\right)^{\ell+1}.
    \]
    The claim of the lemma now follows by applying a union bound over
    all possible choices $v_0,\dots,v_\ell$ and $T_0,\dots,T_\ell$.
\hspace*{\fill}\qed\end{proof}

Based on Lemma \ref{lemma:cliques}, we obtain the following theorem.

\begin{theorem}\label{thm:multipletoken}
    On always connected dynamic networks with $k$ tokens in which
    nodes initially know at most $k/2$ tokens on average, every
    centralized randomized token-forwarding algorithm requires at least
    \[
    \Omega\left( \frac{nk}{(\log n + b\log k)b\log\log b} \right) \geq 
    \Omega\left( \frac{nk}{b^2 \log n \log \log n} \right)
    \]
    rounds to disseminate all tokens to all nodes.
\end{theorem}

\begin{proof}
    For $w_i=2^i$, let $\ell_i+1$ be the size of the largest set
    $S_i$, such that that edge between any two nodes $u,v\in S_i$ has
    weight at least $w_i$. Hence, in the MST, there are at most
    $\ell_i$ edges with weight between $w_i$ and $2w_i$. The amount by
    which the potential function $\Phi$ increases in round $r$ can
    then be upper bounded by
    \[
    \sum_{i=0}^{\log b} 2w_i\cdot\ell_i = \sum_{i=0}^{\log b}2^{i+1}\cdot\ell_i.
    \]
    By Lemma \ref{lemma:cliques} (and a union bound over the $\log b$
    different $w_i$), for a sufficiently small constant $\eps>0$,
    \[
    \ell_i = O\left(\frac{\log n + b\log k}{w_i\log w_i}\right) =
    O\left(\frac{\log n + b\log k}{2^i\cdot i}\right)
    \]
    with high probability. The number of tokens learned in each round
    can thus be bounded by
    \[
    \sum_{i=0}^{\log b}O\left(\frac{\log n + b\log k}{i}\right) = 
    O\big((\log n + b\log k)\log\log b\big).
    \]
    By a standard Chernoff bound, with high probability, the initial
    potential is of the order $1-\Theta(nk/b)$. Therefore to
    disseminate all tokens to all nodes, the potential has to increase
    by $\Theta(nk/b)$ and the claim follows.
\hspace*{\fill}\qed\end{proof}

\subsection{Interval Connected Dynamic Networks}
\label{sec:interval}

While allowing that the network can change arbitrarily from round to
round is a clean and useful theoretical model, from a practical point
of view it might make sense to look at dynamic graphs that are a bit
more stable. In particular, some connections and paths might remain
reliable over some period of time. In \cite{KLO}, token dissemination
and the other problems considered are studied in the context of
$T$-interval connected graphs. For $T$ large enough, sufficiently many
paths remain stable for $T$ rounds so that it is possible to use
pipelining along the stable paths to disseminate tokens significantly
faster (note that this is possible even though the nodes do not know
which edges are stable). In the following, we show that the lower
bound described in Section \ref{sec:original} can also be extended to
$T$-interval connected networks.

\begin{theorem}
    On $T$-interval connected dynamic networks in which nodes
    initially know at most $k/2$ of $k$ tokens on average, every
    randomized token-forwarding algorithm requires at least
    \[
    \Omega\left( \frac{nk}{T(T\log k + \log n)} \right) \geq 
    \Omega\left( \frac{nk}{T^2 \log n} \right) 
    \]
    rounds to disseminate all tokens to all nodes.
\end{theorem}
\begin{proof}
    We assume that each of the sets $K'_u(0)$ independently contains
    each of the $k$ tokens with probability $p=1-\eps/T$ for a
    sufficiently small constant $\eps>0$. As before, we let $i_u(r)$
    be the token broadcast by node $u$ in round $r$ and call the set
    of pairs $(u,i_u(r))$ the token assignment of round $r$. In the
    analysis, we will also make use of token assignments of the form
    $\calT=\set{(u,I_u): u\in V}$, where $I_u$ is a set of tokens sent by
    some node $u$.
    
    Given a token assignment $\calT=\set{(u,I_u)}$, as in the previous
    subsection, an edge $\set{u,v}$ is free in particular if $I_u\!\subseteq\!
    K_v'(0)\land I_v\!\subseteq\! K_u'(0)$. Let $E_\calT$ be the free edges
    w.r.t.\ a given token assignment $\calT$.  Further, we define
    $\calS_\calT=\set{S_{\calT,1},\dots,S_{\calT,\ell}}$ to be the
    partition of $V$ induced by the components of the graph
    $(V,E_\calT)$.

    Consider a sequence of $2T$ consecutive rounds
    $r_1,\dots,r_{2T}$. For a node $v_j$ and round $r_i$, $i\in[2T]$, let
    $I_{i,j}:=\set{i_{v_j}(r_1),\dots,i_{v_j}(r_i)}$ be the set of
    tokens transmitted by node $v_j$ in rounds $r_1,\dots,r_i$ and let
    $\calT_i:=\set{(v_1,I_{i,1}),\dots,(v_n,I_{i,n})}$. As above, let
    $E_{\calT_i}$ be the free edges for the token assignment $\calT_i$
    and let $\calS_{\calT_i}$ be the partition of $V$ induced by the
    components of the graph $(V,E_{\calT_i})$. Note that for $j>i$,
    $E_{\calT_j}\subseteq E_{\calT_i}$ and $\calS_{\calT_j}$ is a
    sub-division of $\calS_{\calT_i}$.

    We construct edge sets $E_1,\dots,E_{2T}$ as follows. The set
    $E_1$ contains $|\calS_{\calT_1}|-1$ edges to connect the
    components of the graph $(V,E_{\calT_1})$. For $i>1$, the edge set
    $E_i$ is chosen such that $E_i\subseteq E_{\calT_{i-1}}$,
    $|E_i|=|\calS_{\calT_{i}}|-|\calS_{\calT_{i-1}}|$, and the graph
    $(V,E_{\calT_i}\cup E_1\cup\dots\cup E_i)$ is connected. Note that
    such a set $E_i$ exists by induction on $i$ and because
    $\calS_{\calT_i}$ is a sub-division of $\calS_{\calT_{i-1}}$.

    For convenience, we define
    $E_{\set{r_1,\dots,r_i}}:=E_1\cup\dots\cup E_i$. By the above
    construction, the number of edges in $E_{\set{r_1,\dots,r_i}}$ is
    $|\calS_{\calT_{i}}|-1$, where $|\calS_{\calT_i}|$ is the number
    of components of the graph $(V,E_{\calT_i})$. Because in each
    round, every node transmits only one token, the number of tokens
    in each $I_{i,j}\in \calT_i$ is at most $|I_{i,j}|\leq i\leq 2T$.
    By Lemma \ref{lemma:cliques}, if the constant $\eps$ is chosen
    small enough, the number of components of $(V,E_\calT)$ and
    therefore the size of $E_{\set{r_1,\dots,r_i}}$ is upper bounded
    by $|\calS_\calT|\leq \log n + T\log k$, w.h.p.

    We construct the dynamic graph as follows. For simplicity, assume
    that the first round of the execution is round $0$. Consider some
    round $r$ and let $r_0$ be the largest round number such that
    $r_0\leq r$ and $r_0\equiv 0\pmod{T}$. The edge set in round $i$
    consists of the the free edges in round $i$, as well as of the
    sets $E_{i_0-T,\dots,i}$ and $E_{i_0,\dots,i}$. The resulting
    dynamic graph is $T$-interval-connected. Furthermore, the number
    of non-free edges in each round is $O(\log n + T\log k)$. Because
    in each round, at most $2$ tokens are learned over each non-free
    edge, the theorem follows.
\hspace*{\fill}\qed\end{proof}

\subsection{Vertex Connectivity}
\label{sec:connectivity}

Rather than requiring more connectivity over time, we now consider the
case when the network is better connected in every round. If the
network is $c$-vertex connected for some $c>1$, in every round, each
set of nodes can potentially reach $c$ other nodes (rather than just
$1$). In \cite{KLO}, it is shown that for the basic greedy token
forwarding algorithm, one indeed gains a factor of $\Theta(c)$ if the
network is $c$-vertex connected in every round. We first need to state
two general facts about vertex connected graphs.

\begin{proposition}\label{lem:cconnect}
If in a graph $G$ there exists a vertex $v$ with degree at least $c$ such that $G - \set{v}$ is $c$-vertex connected then $G$
is also $c$-vertex connected. 
\end{proposition}


\begin{lemma}\label{lem:kconnectmindegree}
    For $c$, any $n$-node graph $G=(V,E)$ with minimum degree at least
    $2c-2$ can be augmented by $n$ edges to be $c$-vertex connected.
\end{lemma}
 \begin{proof}
     We specialize the much more powerful results of
     \cite{jackson2005independence} which characterize the minimum
     number of augmentation edges needed to our setting:

     According to \cite[p41, criterion 4]{jackson2005independence} any
     graph with minimum degree at least $2c-2$ is $c$-independent and
     for such a graph $G$ it holds that the minimum number of edges
     needed to make it $c$-vertex connected is exactly $\max\{b_c(G)-1,
     \ceil{t_c(G)/2}\}$ \cite[Theorem
     3.12]{jackson2005independence}. Here, $b_c(G)$ is the maximum
     number of connected components $G$ can be dissected by removing
     $c-1$ nodes (which is at most $n-c+1$) and $t_c$ is at most the
     maximum value for $\sum_i c - |\Gamma(X_i)|$ that can be obtained
     for a disjoint node partitioning $X_1, X_2, \ldots, X_p$
     \cite[p33]{jackson2005independence}. Here $\Gamma(X_i)$ is the
     set of nodes neighboring $X_i$, i.e., $\Gamma(X_i)=\set{v\in
         V\setminus X_i: \exists u \in X_i \text{ s.t.\ } \set{u,v}\in
         E}$.  Because every node has degree at least $2c-2$,
    $|\Gamma(X_i)| \geq (2c-2) - |X_i| + 1$ and thus
     \[
     \sum_i c - |\Gamma(X_i)| \leq 
     \sum_i c - (2c-2) - |X_i| + 1 = \sum_i |X_i| - \sum_i (c-1) \leq
     n. \vspace{-0.5cm}
     \]
 \hspace*{\fill}\qed\end{proof}

 We will also need the following basic result about weighted sums of
 Bernoulli random variables.

\begin{lemma}\label{lem:increase}
    For some $c$ let $\ell_1, \ell_2, \ldots, \ell_\tau$ be positive integers with
    $\ell = \sum_i \ell_i > c$.  Furthermore, let $X_1, X_2, \ldots, X_\tau$
    be i.i.d.\ Bernoulli variables with $\Pr[X_i = 1] = \Pr[X_i = 0] =
    1/2$ for all $i$.  For any integer $x > 1$ it holds that:
    \[
    \Pr\left[ \min \set{|L|\ :\ L\subseteq[\tau]\land 
            \sum_{i \in \{j | X_j = 1\} \cup L} \ell_i \geq
        c} > x \right] < 2^{-\Theta(\frac{x\ell}{c})}.
    \]
    That is, the probability that $x$ of the random variables need to
    be switched to one after a random assignment in order get $\sum_i
    X_i \ell_i \geq c$ is at most $2^{-\Theta(\frac{x\ell}{c})}$.
\end{lemma}

 \begin{proof}
     Fix a positive integer $x$. Suppose without loss of generality
     that $\ell_1 \geq \ell_2 \geq \ldots \geq \ell_l$. Clearly $\min
     \{|L|\ | \ \sum_{i \in \{j | X_j = 1\} \cup L} \ell_i \geq c \}
     \leq x$ always holds if $\sum_{i \leq x} \ell_i \geq c$. Thus,
     there is nothing to show unless $\ell_i \leq c/x$ for all $i \geq
     x$ and $\sum_{i>x} \ell_i > \frac{2}{3}\ell$. For this case,
     consider a scaling by a factor of $\frac{x}{c}$ of all the
     values. The scaled values $\ell_i \frac{x}{c}$ for $i \geq x$ are
     at most one and the scaled expectation of the sum is
     $E\big[\sum_{i>x} X_i \cdot\big(\ell_i \frac{x}{c}\big) \big] >
     \frac{2}{3}\ell\cdot\frac{x}{c}\cdot\frac{1}{2}=\frac{\ell
         x}{3c}$. A standard Chernoff bound then shows that $\Pr\big[
     \sum_{i>x} X_i \ell_i < c \big] < 2^{-\Theta(\frac{x\ell}{c})}$.
 \hspace*{\fill}\qed\end{proof}

To prove our lower bound for always $c$-connected graph, we initialize
the $K'$-sets as for always connected graphs, i.e., each token $i$ is
contained in every set $K'_u(0)$ with constant probability $p$ (we
assume $p=1/2$ in the following). In each
round, the adversary picks a $c$-connected graph with as few free
edges as possible. Using Lemmas \ref{lem:cconnect} and
\ref{lem:kconnectmindegree}, we will show that a graph with a small
number of non-free edges can be constructed as follows. First, as long
as we can, we pick vertices with at least $c$ neighbors among the
remaining nodes. We then show how to extend the resulting graph to a
$c$-connected graph.

\begin{lemma}\label{lem:cconLB1}
    With high probability (over the choices of the sets $K'_u(0)$),
    for every token assignment $(u,I_u(r))$, the largest set $S$ for
    which no node $u\in S$ has at least $c$ neighbors in $S$ is of
    size $O(c\log n)$.
\end{lemma}
\begin{proof}
    Consider some round $r$ with token assignment
    $\set{\big(u,i_u(r)\big)}$. We need to show that for any set $S$
    of size $s=\alpha c\log n$ for a sufficiently large constant
    $\alpha$, at least one node in $S$ has at least $c$ free neighbors
    in $S$ (i.e., the largest degree of the graph induced by the free
    edges between nodes in $S$ is at least $c$). 

    We will use a union bound over all $n^s$ sets $S$ and all $k^s$
    possibilities for selecting the tokens sent by these nodes. We
    want to show that if the constant $\alpha$ is chosen sufficiently
    large, for each of these $2^{s \log nk}$ possibilities we have a
    success probability of at least $1 - 2^{-2s \log nk}$.

    We first partition the nodes in $S$ according to the token sent
    out, i.e., $S_i$ is the subset of nodes sending out token
    $i$. Note that if for some $j$ we have $S_j > c$ we are done since
    all edges between nodes sending the same token are free. With
    this, let $j^*$ be such that $\sum_{i<j^*} |S_i| \geq s/3$ and
    $\sum_{i>j^*} |S_i| \geq s/3$. We now claim that for every $j <
    j^*$, with probability at most $2^{-6|S_j|\log nk}$, there does
    not exist a node in $S_j$ that has at least $c$ free edges to
    nodes in $S' = \bigcup_{i>j^*} S_j$. Note that the events that a
    node from $S_j$ has at least $c$ free edges to nodes in $S'$ are
    independent for different $j$ as it only depends on which nodes
    $u$ in $S'$ have $j$ in $K'_u(0)$ and on the $K'(0)$-sets of
    the nodes in $S_j$. The claim that we have a node with degree $c$
    in $S$ with probability at least $1-2^{-2s\log nk}$ then follows
    from the definition of $j^*$.

    Let us therefore consider a fixed value $j$.  We first note that
    for a fixed $j$ by standard Chernoff bounds with probability at
    least $1-2^{-\Theta(s)}$, there at least $s/3 \cdot p/2=s/12$
    nodes in $S'$ that have token $j$ in their initial $K'$-set. For
    $\alpha$ sufficiently large, this probability is at least
    $1-2^{-7c \log nk} \geq 1-2^{-7|S_j| \log nk}$. In the following,
    we assume that there are at least $s/12$ nodes $u$ in $S'$ for
    which $j\in K'_u(0)$.

    Let $s_{j,i}$ for any $i>j^*$ denote the number of nodes in $S_i$
    that have token $j$ in the initial $K'$-set. The number of free
    edges to a node $u$ in $S_j$ is at least $\sum_{i > j^*} X_{u,i}
    s_{i,j}$, where the random variable $X_{u,i}$ is $1$ if node $u$
    initially has token $i$ in $K_u'(0)$ and $0$ otherwise (i.e.,
    $X_{u,i}$ is a Bernoulli variable with parameter $1/2$). Note that
    since $\sum_i s{j,i}\geq s/12$, the expected value of the number
    of free edges to a node $u$ in $S_j$ is at least $s/24$. By a
    Chernoff bound, the probability that the number of free edges from
    a node $u$ in $S_j$ does not deviate by more than a constant
    factor with probability $1-2^{-\Theta(s/c)}$. Note that
    $s_{j,i}\leq c$ since $|S_j|\leq c$.  For $\alpha$ large enough
    this probability is at least $1-2^{-7\log nk}$. Because the
    probability bound only depends on the choice of $K'_u(0)$, we have
    independence for different $u\in S_j$. Therefore, given that at
    least $s/12$ nodes in $S'$ have token $j$, the probability that no
    node in $S_j$ has at least $c$ neighbors in $S'$ can be upper
    bounded as $\big(1-2^{-7|S_j|\log nk}\big)$. Together with the
    bound on the probability that at least $s/12$ nodes in $S'$ have
    token $j$ in their $K'(0)$ set, the claim of the lemma follows.
\hspace*{\fill}\qed\end{proof}

Lemma \ref{lem:cconLB1} by itself directly leads to a lower bound for
token forwarding algorithms in always $c$-vertex connected
graphs. 

\begin{corollary}\label{cor:simplecconn}
    Suppose an always $c$-vertex connected dynamic network with $k$
    tokens in which nodes initially know at most a constant fraction
    of the tokens on average.  Then, any centralized token-forwarding
    algorithm takes at least $\Omega\big(\frac{nk}{c^2\log n}\big)$
    rounds to disseminate all tokens to all nodes.
\end{corollary}
\begin{proof}
    By Lemma \ref{lem:cconLB1}, we know that there exists $K'(0)$-sets
    such that for every token assignment after adding all free edges,
    the size of the largest induced subgraph with maximum degree less
    than $c$ is $O(c\log n)$. By Lemma \ref{lem:cconnect}, it suffices
    to make the graph induced by these $O(c\log n)$ nodes $c$-vertex
    connected to have a $c$-vertex connected graph on all $n$
    nodes. To achieve this, by Lemma \ref{lem:kconnectmindegree}, it
    suffices to increase all degrees to $2c-2$ and add another
    $O(c\log n)$ edges. Overall, the number of non-free edges we have
    to add for this is therefore upper bounded by $O(c^2\log
    n)$. Hence, the potential function increases by at most $O(c^2\log
    n)$ per round and since we can choose the $K'(0)$-sets so that
    initially the potential is at most $\lambda nk$ for a constant
    $\lambda<1$, the bound follows.
\hspace*{\fill}\qed\end{proof}

As shown in the following, by using a more careful analysis, we can
significantly improve this lower bound for $c=\omega(\log n)$.  Note
that the bound given by the following theorem is at most an $O(\log^{3/2}
n)$ factor away from the simple ``greedy'' upper bound.

\begin{theorem}\label{lem:cconLB2}
    Suppose an always $c$-vertex connected dynamic network with $k$
    tokens in which nodes initially know at most a constant fraction
    of the tokens on average.  Then, any centralized token-forwarding
    algorithm takes at least $\Omega\big(\frac{nk}{c\log^{3/2} n}\big)$
    rounds to disseminate all tokens to all nodes.
\end{theorem}
\begin{proof}
    We use the same construction as in \Cref{lem:cconLB1} to obtain a
    set $S$ of size $|S| = s = \alpha c \log n$ for a sufficiently
    large constant $\alpha>0$ such that $S$ needs to be augmented to a
    $c$-connected graph. Note that we want the set to be of size $s$
    and therefore we do not assume that in the induced subgraph, every
    node has degree less than $c$. We improve upon \Cref{lem:cconLB1}
    by showing that it is possible to increase the potential function
    by adding a few more tokens to the $K'$-sets, so that afterwards
    it is sufficient to add $O(s)$ additional non-free edges to $S$ to
    make the induced subgraph $c$-vertex connected.  Hence, an
    important difference is that are not counting the number of edges
    that we need to add but the number of tokens we need to give away
    (i.e., add to the existing $K'$-sets).

    We first argue that w.h.p., it is possible to raise the minimum
    degree of vertices in the induced subgraph of $S$ to $2c$ without
    increasing the potential function by too much. Then we invoke
    \Cref{lem:kconnectmindegree} and get that at most $O(s)$ more edges
    are then needed to make $S$ induce a $c$-connected graph as
    desired.

    We partition the nodes in $S$ according to the token sent out in
    the same way as in the proof of \Cref{lem:cconLB1}, i.e., $S_i$ is
    the subset of nodes sending out token $i$. Let us first assume
    that no set $S_i$ contains more than $s/3$ nodes. We can then
    divide the sets of the partition into two parts with at least
    $s/3$ nodes each. To argue about the sets, we rename the tokens
    sent out by nodes in $S$ as $1,2,\dots$ so that we can find a
    token $j^*$ for which $\sum_{j=1}^{j^*}|S_j|\geq s/3$ and
    $\sum_{j>j^*}|S_j|\geq s/3$. We call the sets $S_j$ for $j\leq
    j^*$ the left side of $S$ and the sets $S_j$ for $j>j^*$ the right
    side of $S$. If there is a set $S_i$ with $|S_i|>s/3$, we define
    $S_i$ to be the right side and all other sets $S_j$ to be the left
    side of $S$. We will show that we can increase the potential
    function by at most $O(s\sqrt{\log n})=O(c\log^{3/2}n)$ such that
    all the nodes on the left side have at least $2c$ neighbors on the
    right side. If all sets $S_i$ are of size at most $s/3$,
    increasing the degrees of the nodes on the right side is then done
    symmetrically. If the right side consists of a single set $S_i$ of
    size at least $s/3$, for $\alpha$ large enough we have $s/3\geq
    2c+1$ and therefore nodes on the right side already have degree at
    least $2c$ by just using free edges.

    We start out by adding some tokens to the sets $K'_u$ for nodes
    $u$ on the right side such that for every token $j\leq j^*$ on the
    left side, there are at least $s/\sqrt{\log n}$ nodes $u$ on the
    right side for which $j\in K'_u$. Let us consider some fixed token
    $j\leq j^*$ from the left side. Because every node $u$ on the
    right side has $j\in K'_u(0)$ with probability $1/2$, with
    probability at least $1 - 2^{-\Theta(s)}$, at least $s/\sqrt{\log
        n}$ nodes u on the right side have $j\in K'_u(0)$. For such a
    token $j$, we do not need to do anything. Note that the events
    that $j\in K'_u(0)$ are independent for different $j$ on the left
    side. Therefore, for a sufficiently large constant $\beta$ and a
    fixed collection of $\beta\log n$ tokens $j$ sent by nodes on the
    left side, the probability that none of these tokens is in at
    least $s/\sqrt{\log n}$ sets $K'_u(0)$ for $u$ on the right side
    is at most $2^{-\gamma s\log n}$ for a given constant
    $\gamma>0$. As there are at most $s$ tokens sent by nodes on the
    left side, the number of collections of $\beta\log n$ tokens is at
    most
    \[
    {s\choose \beta\log n} \leq 
    \left(\frac{es}{\beta\log n}\right)^{\beta\log n} =
    \left(\frac{e\alpha c}{\beta}\right)^{\beta\log n} =
    2^{\Theta(\log c \log n)},
    \]
    which is less than $2^{s\log n}$ for sufficiently large
    $\alpha$. Hence, with probability at least $1-2^{-(\gamma-1)s\log n}$, for
    at most $\beta\log n$ tokens $j$ on the left side there are less
    than $s/\sqrt{\log n}$ nodes $u$ on the right that have $j\in
    K'_u(0)$.  For these $O(\log n)$ tokens $j$, we add to $j$ to
    $K'_u$ for at most $s/\sqrt{\log n}$ nodes $u$ on the right side,
    such that afterwards, for every token $j$ sent by a node on the
    left side, there are at least $s/\sqrt{\log n}$ nodes $u$ on the
    right for which $j\in K'_u$. Note that this increases the
    potential function by at most $O(s\sqrt{\log n}) =
    O(c\log^{3/2}n)$.

    We next show that by adding another $O(c\log^{3/2}n)$ tokens to
    the $K'$-sets of the nodes on the left side, we manage to get that
    every node $u$ on the left side has at least $2c$ free neighbors
    on the right side. For a token $j\leq j^*$ sent by some node  on
    the left side and a token $i>j^*$ sent by some node  on the right
    side, let $s_{i,j}$ be the number of nodes $u\in S_i$ for which
    $j\in K'_u$. Note that if token $i$ is in $K'_v$ for some $v\in
    S_j$, $v$ has $s_{i,j}$ neighbors in $S_i$.

    Using the augmentation of the $K'_u$-sets for nodes on the right,
    we have that for every $j\leq j^*$, $\sum_{i>j^*} s_{i,j}\geq
    s/\sqrt{\log n}$. For every $i>j^*$, with probability $1/2$, we
    have $i\in K'_v(0)$. In addition, we add tokens additional $i$ to
    $K'_v$ for which $i\not\in K'_v(0)$ such that in the end,
    $\sum_{i>j^*, i\in K'_v} s_{i,j}\geq 2c$. By Lemma
    \ref{lem:increase}, the probability that we need to add $\geq x$ tokens
    is upper bounded by $2^{-\Theta\left(xs/(c\sqrt{\log
                n})\right)}=2^{-\Theta\left(x\sqrt{\log
                n}\right)}$. As the number of tokens we need to add to
    $K_v'$ is independent for different $v$, in total we need to add
    at most $O\big(\frac{s\log n}{\sqrt{\log n}}\big) =
    O(c\log^{3/2}n)$ tokens with probability at least
    $1-2^{-(\gamma-1)s\log n}$. Note that this is still true after a
    union bound over all the possible ways to distributed the
    $O(c\log^{3/2}n)$ tokens among the $\leq s$ nodes. Using Lemma
    \ref{lem:kconnectmindegree}, we then have to add at most
    $O(s)=O(c\log n)$ additional non-free edges to make the graph
    induced by $S$ $c$-vertex connected.
        
    There are at most $n^s=2^{s\log n}$ ways  to choose the set $S$
    and $k^s=2^{O(s\log n)}$ ways to assign tokens to the nodes in
    $S$. Hence, if we choose $\gamma$ sufficiently large, the
    probability that we need to increase the potential by at most
    $O(c\log^{3/2} n)$ for all sets $S$ and all token assignments is
    positive. The theorem now follows as in the previous lower bounds
    (e.g., as in the proof of Theorem \ref{thm:original}).
\hspace*{\fill}\qed\end{proof}

\section*{Acknowledgments}

We would like to thank Chinmoy Dutta, Gopal Pandurangan, Rajmohan
Rajaraman, and Zhifeng Sun for helpful discussions and for sharing
their work at an early stage.

\clearpage

\bibliographystyle{abbrv}


\end{document}